\documentclass[aps,pra,nopacs,nokeys,superscriptaddress,11pt,twoside,notitlepage,a4paper]{revtex4-1}

\usepackage{graphicx,epic,eepic,epsfig,amsmath,latexsym,amssymb,verbatim,color}
\usepackage{xpatch}

\usepackage{theorem}
\newtheorem{definition}{Definition}
\newtheorem{proposition}[definition]{Proposition}
\newtheorem{lemma}[definition]{Lemma}

\newtheorem{theorem}[definition]{Theorem}

\def\squareforqed{\hbox{\rlap{$\sqcap$}$\sqcup$}}
\def\qed{\ifmmode\squareforqed\else{\unskip\nobreak\hfil
\penalty50\hskip1em\null\nobreak\hfil\squareforqed
\parfillskip=0pt\finalhyphendemerits=0\endgraf}\fi}
\def\endenv{\ifmmode\;\else{\unskip\nobreak\hfil
\penalty50\hskip1em\null\nobreak\hfil\;
\parfillskip=0pt\finalhyphendemerits=0\endgraf}\fi}
\newenvironment{proof}{\noindent \textbf{{Proof~} }}{\qed}
\newenvironment{remark}{\noindent \textbf{{Remark~}}}{\qed}

\mathchardef\ordinarycolon\mathcode`\:
\mathcode`\:=\string"8000
\def\vcentcolon{\mathrel{\mathop\ordinarycolon}}
\begingroup \catcode`\:=\active
  \lowercase{\endgroup
  \let :\vcentcolon
  }

\newcommand{\nc}{\newcommand}
\nc{\rnc}{\renewcommand}
\nc{\beq}{\begin{equation}}
\nc{\eeq}{{\end{equation}}}
\nc{\beqa}{\begin{eqnarray}}
\nc{\eeqa}{\end{eqnarray}}
\nc{\lbar}[1]{\overline{#1}}
\nc{\bra}[1]{\langle#1|}
\nc{\ket}[1]{|#1\rangle}
\nc{\ketbra}[2]{|#1\rangle\!\langle#2|}
\nc{\braket}[2]{\langle#1|#2\rangle}
\nc{\proj}[1]{| #1\rangle\!\langle #1 |}
\nc{\avg}[1]{\langle#1\rangle}
\nc{\Rank}{\operatorname{Rank}}
\nc{\smfrac}[2]{\mbox{$\frac{#1}{#2}$}}
\nc{\tr}{{\operatorname{Tr}\,}}
\nc{\ox}{\otimes}
\nc{\dg}{\dagger}
\nc{\dn}{\downarrow}
\nc{\cA}{{\cal A}}
\nc{\cB}{{\cal B}}
\nc{\cC}{{\cal C}}
\nc{\cD}{{\cal D}}
\nc{\cE}{{\cal E}}
\nc{\cF}{{\cal F}}
\nc{\cG}{{\cal G}}
\nc{\cH}{{\cal H}}
\nc{\cI}{{\cal I}}
\nc{\cJ}{{\cal J}}
\nc{\cK}{{\cal K}}
\nc{\cL}{{\cal L}}
\nc{\cM}{{\cal M}}
\nc{\cN}{{\cal N}}
\nc{\cO}{{\cal O}}
\nc{\cP}{{\cal P}}
\nc{\cQ}{{\cal Q}}
\nc{\cR}{{\cal R}}
\nc{\cS}{{\cal S}}
\nc{\cT}{{\cal T}}
\nc{\cU}{{\cal U}}
\nc{\cX}{{\cal X}}
\nc{\cY}{{\cal Y}}
\nc{\cZ}{{\cal Z}}
\nc{\conv}{{\operatorname{conv}\,}}
\nc{\csupp}{{\operatorname{csupp}}}
\nc{\qsupp}{{\operatorname{qsupp}}}
\nc{\var}{{\operatorname{var}}}
\nc{\rar}{\rightarrow}
\nc{\lrar}{\longrightarrow}
\nc{\polylog}{{\operatorname{polylog}}}
\nc{\1}{{\openone}}
\nc{\wt}{{\operatorname{wt}}}
\nc{\av}[1]{{\left\langle {#1} \right\rangle}}

\nc{\RR}{{{\mathbb R}}}
\nc{\CC}{{{\mathbb C}}}
\nc{\FF}{{{\mathbb F}}}
\nc{\NN}{{{\mathbb N}}}
\nc{\ZZ}{{{\mathbb Z}}}
\nc{\PP}{{{\mathbb P}}}
\nc{\QQ}{{{\mathbb Q}}}
\nc{\UU}{{{\mathbb U}}}
\nc{\EE}{{{\mathbb E}}}
\nc{\id}{{\operatorname{id}}}

\nc{\CHSH}{{\operatorname{CHSH}}}

\nc{\be}{\begin{equation}}
\nc{\ee}{{\end{equation}}}
\nc{\bea}{\begin{eqnarray}}
\nc{\eea}{\end{eqnarray}}
\nc{\<}{\langle}
\rnc{\>}{\rangle}
\nc{\Hom}[2]{\mbox{Hom}(\CC^{#1},\CC^{#2})}
\nc{\rU}{\mbox{U}}

\nc{\ob}[1]{#1}

\nc{\SEP}{{\text{SEP}}}
\nc{\NS}{{\text{NS}}}
\nc{\LOCC}{{\text{LOCC}}}
\nc{\PPT}{{\text{PPT}}}
\nc{\EXT}{{\text{EXT}}}
\nc{\Sym}{{\operatorname{Sym}}}

\nc{\ERLO}{{E_{\text{r,LO}}}}
\nc{\ERLOCC}{{E_{\text{r,LOCC}}}}
\nc{\ERPPT}{{E_{\text{r,PPT}}}}
\nc{\ERLOCCinfty}{{E^{\infty}_{\text{r,LOCC}}}}
\nc{\Aram}{{\operatorname{\sf A}}}

\begin{document}

\title{Energy-constrained diamond norm with applications to the\protect\\ 
       uniform continuity of continuous variable channel capacities}

\date{29 December 2017}

\author{Andreas Winter}
\affiliation{ICREA---Instituci\'o Catalana de Recerca i Estudis Avan\c{c}ats, Pg.~Lluis Companys 23, ES-08010 Barcelona, Spain}
\affiliation{Departament de F\'{\i}sica: Grup d'Informaci\'{o} Qu\`{a}ntica, Universitat Aut\`{o}noma de Barcelona, ES-08193 Bellaterra (Barcelona), Spain.} 
\email{andreas.winter@uab.cat}

\begin{abstract}
The channels, and more generally superoperators acting on the trace class
operators of a quantum system naturally form a Banach space under the
completely bounded trace norm (aka \emph{diamond norm}). 
However, it is well-known that in infinite dimension, the norm topology is 
often ``too strong'' for reasonable applications.
Here, we explore a recently introduced energy-constrained diamond norm
on superoperators (subject to an energy bound on the input states).
Our main motivation is the continuity of capacities and other
entropic quantities of quantum channels, but we also present 
an application to the continuity of one-parameter unitary groups
and certain one-parameter semigroups of quantum channels.
\end{abstract}

\maketitle

\section{Diamond norm}
\label{sec:diamond}
Let $A$ and $B$ be separable Hilbert spaces, in most of the present note
infinite dimensional, of two quantum systems whose states are described by 
the trace class operators $\cT(A)$ and $\cT(B)$, respectively:
\[
  \cT(A) = \{ \xi:A\rightarrow A \text{ s.t. } \|\xi\|_1 < \infty \},
\]
where $\|\xi\|_1 = \tr\sqrt{\xi^\dagger\xi}$ is the trace norm, i.e.~the
sum of all the singular values of $\xi$.
Quantum channels, or in physics language, open system state evolutions,
are modelled as completely positive and trace preserving (cptp) maps
$\cN:\cT(A)\rightarrow\cT(B)$. As these, and more generally their real
linear combinations, which are Hermitian-preserving superoperators 
$\Delta:\cT(A) \longrightarrow \cT(B)$, are linear maps between Banach
spaces, they inherit a natural norm, often called the \emph{trace norm}:
\[\begin{split}
  \|\Delta\|_{1\rightarrow 1} :&= \sup \| \Delta\xi \|_1 \text{ s.t. } \|\xi\|_1 \leq 1  \\
                               &= \sup \| \Delta\rho \|_1 \text{ s.t. } \rho \text{ state on } A.
\end{split}\]
As is well-known, this norm induces a topology on maps that is often too strong 
for quantum mechanical applications -- we will discuss such instances below.
However, as a norm it is actually too weak, since it is not stable under
tensor products with the identity map. Indeed, the natural norm on
superoperators is the of the \emph{completely bounded trace norm}, also known as 
\emph{diamond norm}~\cite{diamond,CB}:
\begin{equation}
  \|\Delta\|_\diamond := \sup \| (\Delta\ox\id_{C})\rho \|_1 \text{ s.t. } \rho \text{ state on } A\ox C.
\end{equation}
From the definition, it is easy to see that the supremum may be restricted 
to pure states (namely by purification and the contractivity of the trace 
norm under partial traces), and that w.l.o.g.~$C = A' \simeq A$. In particular, 
in finite dimension, the supremum is always attained on a pure state on $A\ox A'$.
As a matter of terminology, in the present article, superoperators are generally 
assumed to be Hermitian-preserving, which means that they are
differences of completely positive (cp) maps; 
and \emph{channels} are those maps that are cptp.

The diamond norm has an operational interpretation 
for $\Delta = p\cN_1 - (1-p)\cN_2$, in a Helstrom
context of binary hypothesis testing between to channels $\cN_i$, as follows:
$\frac12 \bigl( 1-\| p\cN_1 - (1-p)\cN_2 \|_{\diamond} \bigr)$ equals the minimum
error probability of distinguishing $\cN_1$ from $\cN_2$, which come with prior
probabilities $p$ and $1-p$, respectively,
when we are allowed preparation of a probe state $\rho^{AC}$, one application of 
the unknown channel, and an arbitrary measurement on the system $BC$.

\medskip
The diamond norm not only gives an important and natural metric on quantum
channels \cite{diamond,CB,Watrous:diamond}, it turns out that it provides also
the natural setting to discuss continuity of channel capacities in finite 
dimension. This proceeds via the Alicki-Fannes inequality \cite{AlickiFannes}
and the ``telescoping'' trick of Leung and Smith \cite{LeungSmith:cont}.
Here, we derive inequalities motivated by a desire to generalise the latter
to infinite dimensional (for instance Bosonic) channels. The entropy is
known to enjoy Fannes-type continuity subject to energy bounds \cite{Winter-S}, 
and so we have to look at the correct metric on channels to make the
Leung-Smith argument work.

The rest of the paper is organised as follows: In section~\ref{sec:issues}
we review some of the issues with the diamond norm in infinite dimensional
systems; then in section~\ref{sec:e-bounded-diamond} we present the 
energy-constrained diamond norm and review some of its immediate mathematical 
properties, as well as presenting some less obvious ones.
After that we move to two groups of applications: 
First (section~\ref{sec:continuous-time}), we show that the energy constraint
allows us to recover the norm continuity of unitary time evolution 
w.r.t.~the generating Hamiltonian, and certain one-parameter semigroups;
we highlight a connection of the explicit norm continuity bounds to the
topic of quantum speed limits.
Secondly, in section~\ref{sec:continuity-S}, we apply the formalism to prove
a uniform continuity lemma for the conditional entropy of quantum channel
outputs under an energy constraint at the input, and use it to prove uniform 
continuity of channel capacities with energy constraints at the sender.

\section{Undesirable features of the diamond norm in infinite dimension}
\label{sec:issues}
It is well-known that the superoperator trace norm or diamond norm
convergence is not suitable to make one-parameter semigroups
$t \longmapsto \Lambda(t) = e^{t\cL}$ continuous, unless the generator
$\cL$ is bounded~\cite{Lindblad,SHW:lindblad}. In fact, from the theory
of one-parameter unitary groups, i.e.~free quantum mechanical time
evolution, we know that the natural, correct notion of continuity in
time is provided by the strong topology of pointwise convergence.
On the other hand, strong convergence is not defined in terms of
a single norm, but in certain applications we need a quantitative
measure of how far an element of a sequence is from its limit,
in other words: a metric.
To illustrate the issues in some concrete examples, start by
considering the single-mode Bosonic quantum limited attenuators $\cA_\eta$, 
$|\eta|\leq 1$, defined uniquely by the property that for all coherent 
states $\ket{\alpha}$,
\[
  \cA_\eta(\proj{\alpha}) = \proj{\eta\alpha},
\]
which can be realised using a beam splitter with transmissivity $\eta$,
whose one input channel is fed with the input state and the other
with the vacuum state $\ket{0}$.

\begin{proposition}
  \label{prop:attenuators}
  For any $\eta \neq \eta'$, we have 
  $\| \cA_\eta - \cA_{\eta'} \|_\diamond = 2$. 
\end{proposition}  
I.e., all attenuators are at maximum distance from each other, 
despite the $(\cA_{e^{-t}})_{t\in\RR}$ forming a seemingly 
``continuous'' one-parameter semigroup; which it is, but not with respect
to the diamond norm topology, but the so-called strong topology 

\medskip
\begin{proof}
Clearly, $\| \cN_1 - \cN_2 \|_\diamond \leq 2$ for any two channels.
On the other hand, for $\eta \neq \eta'$,
\[\begin{split}
  \| \cA_\eta - \cA_{\eta'} \|_\diamond 
           &\geq \sup_\alpha \| \cA_\eta(\proj{\alpha}) - \cA_{\eta'}(\proj{\alpha}) \|_1 \\
           &=    \sup_\alpha \| \proj{\eta\alpha} - \proj{\eta'\alpha} \|_1   \\
           &=    \sup_\alpha 2\sqrt{1-|\bra{\eta\alpha}\eta'\alpha\rangle|^2} \\
           &=    \sup_\alpha 2\sqrt{1-\exp\left(-|(\eta-\eta')\alpha|^2\right)}
            =    2,
\end{split}\]
where in the first two lines we have simply inserted coherent states $\ket{\alpha}$
as test stated and evaluated the channels on those; in the third line we have
used the well-known relation between the trace distance and inner product of
pure states, and in the fourth the equally well-known formula for the
inner product of coherent states.
\end{proof}

\medskip
What goes on here is that to realise this large distance between
different attenuator channels, we need to probe them with highly
energetic test states, albeit only coherent ones in the example. 
This is at odds with most communication settings where continuous
variable channels are used under an energy constraint on the input.
The same is true for one-mode squeezing unitaries, or displacement unitaries; 
the corresponding channels are always at mutual diamond norm distance $2$.
A yet more fundamental example are time evolutions generated by unbounded
Hamiltonians:

\begin{proposition}
  \label{prop:time-evolution}
  For an unbounded Hamiltonian $H\geq 0$, let $U_t = e^{-itH}$ and
  $\cU_t(\rho) = U_t \rho U_t^\dagger$ be the time-propagator. 
  Then there is a dense set $\mathcal{D}\subset\RR$ 
  such that $\| \cU_t - \cU_{t'} \|_\diamond = 2$ for all times
  $t,t'$ with $t-t'\in\mathcal{D}$.
\end{proposition}  
\begin{proof}
By the one-parameter group property of $\cU(t)$, and the unitary
invariance of the diamond norm, we only have to prove the statement
for $t'=0$. 
We may w.l.o.g.~assume that the smallest eigenvalue of $H$ is $0$.
We shall give a simple description of a possible set $\cD$ in 
terms of the Hamiltonian's spectrum $S=\operatorname{spec}\,H$: 
\[
  \cD = \left\{ t : 0 \in \overline{\conv(\operatorname{spec}\,U_t)} 
                           = \overline{\conv\{e^{-itE}:E\in S\}} \right\}.
\]
Indeed, for a unit vector in eigenstate superposition 
$\ket{\psi} = \sum_E c_E \ket{E}$, we have 
\[
  \bra{\psi} e^{-itH} \ket{\psi} = \sum_E |c_E|^2 e^{-itE},
\]
which can be made $0$ by choosing the probabilities $|c_E|^2$
if and only if the convex hull of the $e^{-itE}$ contains the origin. 
It can be made arbitrarily small if the closure of the convex hull
contains the origin. In other words, $\ket{\psi}$ and $e^{-itH} \ket{\psi}$
can be made arbitrarily close to orthogonal.

To show that $\cD$ is dense in $\RR$, consider a point $t_0\not\in\cD$.
This means that there is a line separating the convex set
$\cC_0 := \overline{\conv\{e^{-itE}:E\in S\}}$ from the origin:
\[
  \forall z \in \cC_0 \quad \operatorname{Re}\,z e^{i\alpha} \geq \delta > 0.
\]
Now, consider any sufficiently large $E_0 > 0$; then for all sufficiently
small $\Delta t > 0$,
\[
  \forall t\in [t_0-\Delta t,t_0+\Delta t] \ 
  \forall E \in S \cap [0;E_0] \quad \operatorname{Re}\,e^{-itE} e^{i\alpha} \geq \frac12 \delta.
\]
On the other hand, by the unboundedness of $H$, there must exist an 
energy $E \in S$, $E_1 > E_0$, such that $2\Delta t E_1 \geq 2\pi$, meaning
that $e^{-itE_1}$ makes at least one full revolution on the unit circle
as $t$ varies from $t_0-\Delta t$ to $t_0+\Delta_t$.
By the intermediate value theorem, there must be at least one $t$ in the
interval such that $0\in\overline{\conv(S\cap[0;E_0] \cup \{E_1\})}$.
Since $\Delta t$ was arbitrarily small, this shows that arbitrarily 
close to $t_0$ elements of $\cD$ can be found.
\end{proof}

\medskip
\begin{remark}
Typical Hamiltonians will have times where $\| \cU(t)-\id \|_\diamond < 2$.
A simple example is given by the quantum harmonic oscillator, 
$H = \sum_{n\geq 0} \left(n+\frac12\right)$. Because of the integer
spacing of the eigenenergies, it has period $\tau = 2\pi$,
in particular $\cU(\tau) = \id = \cU(0)$.

One can also construct Hamiltonians with infinitely many exceptional
times. Consider for instance $H = \sum_{n\geq 0} c^n \proj{n}$, with
and integer $c\geq 3$, which for times $\tau_k = 2\pi c^{-k}$ gives
rise to
\[
  U_{\tau_k} = \sum_{n\geq 0} e^{2\pi i c^{n-k}} \proj{n}
             = \sum_{n=0}^{k-1} e ^{2\pi i c^{n-k}} \proj{n} + \sum_{n\geq k} \proj{n}.
\]
This differs from $\1$ only on the sector of the $k$ smallest
energies, and is indeed arbitrarily close to $\1$ for sufficiently
large $c$.

We conjecture that periods affecting all but a finite number of
the eigenfrequencies are the only possible origin of exceptional times. 
Concretely, let $E_0$ be the ground state energy of $H$, and define the essential
periods as
\[
  \cP := \{ t : tE - tE_0 \in \ZZ \text{ for all but finitely many } E\in\operatorname{spec}\,H \}.
\]
Note that $\cP$ is an at most countably infinite set, and may be empty.
Then, it seems plausible that for all $t-t'\not\in\cP$, it holds
$\| \cU_t - \cU_{t'} \|_\diamond = 2$.
\end{remark}

\section{A diamond norm relative to a Hamiltonian and an energy bound.}
\label{sec:e-bounded-diamond}
Assume that system $A$ comes with a genuine Hamiltonian $H_A \geq 0$,
in fact normalised in such a way that the ground state energy is $0$
(we will call such Hamiltonians \emph{grounded}).
In particular, we want the pure state vectors $\ket{\psi}$ of finite
energy, $\bra{\psi} H_A \ket{\psi} < \infty$, to be dense in $A$. This
implies that the subspace spanned by those vectors with 
$\bra{\psi} H_A \ket{\psi} \leq E$ is also dense in $A$, for any $E>0$.

There are two stronger conditions that are useful sometimes:
First, it may be that $H_A$ has discrete spectrum and each eigenvalue has
only finite degeneracy; equivalently, this can be expressed as saying
that the projector $\{ H_A \leq E \}$ onto the eigenvalues up to $E$
has finite rank for all $E\geq 0$. We will call such Hamiltonians
simply \emph{discrete}.
Second, for later use with von Neumann entropies, Hamiltonians will at certain
moments also be required to satisfy the \emph{Gibbs hypothesis}, meaning that for all
$\beta > 0$, $Z = \tr e^{-\beta H_A} < \infty$, so that the Gibbs state
$\gamma(E) = \frac1Z e^{-\beta H_A}$, with a certain function $\beta=\beta(E)$,
is well-defined and characterised as the unique entropy maximiser subject 
to the energy bound $\tr \rho H_A \leq E$. Such Hamiltonians are always discrete,
but the converse is not true. If we have different Hamiltonians
$H_A$, $H_B$, etc, to distinguish, we denote the Gibbs state as $\gamma_A(E)$,
$\gamma_B(E)$, etc, for clarity.

\begin{definition}[Shirokov~\cite{Shirokov:E-diamond}]
  \label{def:diamond-E}
  For a Hermitian-preserving map $\Delta$, define the 
  \emph{$E$-constrained diamond norm}
  \begin{equation}
    \|\Delta\|_{\diamond E} := \sup \| (\Delta\ox\id_{C})\rho \|_1 
                                 \text{ s.t. } \rho \text{ state on }A\ox C,\ 
                                               \tr\rho^A H_A \leq E.
  \end{equation}
\end{definition}
By the same reasoning as for the diamond norm, the supremum can
be restricted to pure states, and w.l.o.g.~$C = A' \simeq A$.
A related definition was proposed by Pirandola \emph{et al.}~\cite{PLOB},
for the special case of quantum harmonic oscillators, and with the slight 
difference that the energy (photon number) bound was applied to both
system $A$ and reference $C$. The resulting norm is equivalent to the
one of Shirokov, but because of the properties proved in the following,
and the particular applications we have in mind, we have chosen the
latter.

That this definition is indeed a norm, and some of its elementary properties
(several of which noted already in \cite{Shirokov:E-diamond}), 
are contained in the following lemma.

\begin{lemma}
  \label{lemma:diamond-E-norm}
  For a grounded system Hamiltonian $H_A \geq 0$ and all $E>0$:
  \begin{enumerate}
    \item $\|\cdot\|_{\diamond E}$ is a norm on superoperators, in particular
      $\|\Delta\|_{\diamond E} = 0$ iff $\Delta=0$.

    \item For every $\Delta$, $\|\Delta\|_{\diamond E}$ is monotonically
      non-decreasing and concave in $E$. In particular, for $E < E'$,
      \[
        \|\Delta\|_{\diamond E} \leq \|\Delta\|_{\diamond E'} \leq \frac{E'}{E}\|\Delta\|_{\diamond E}.
      \]
      Thus, all $\|\cdot\|_{\diamond E}$ are topologically equivalent.

    \item The diamond norm is obtained as a limit:
      $\displaystyle{\sup_{E>0}} \|\Delta\|_{\diamond E} = \|\Delta\|_{\diamond}$.

    \item An alternative, and sometimes more convenient formula for the norm is
      \[
        \| \Delta \|_{\diamond E} = \sup_{\xi} \| (\Delta\ox\id)\xi \|_1 
                                  = \sup_{\xi,\, T} \tr [(\Delta\ox\id)\xi]T,
      \]
      where the suprema are over $\xi = \rho-\sigma$ with $\rho,\,\sigma \geq 0$,
      $\tr \rho + \tr \sigma \leq 1$ and $\tr\rho^A H_A,\,\tr\sigma^A H_A \leq E$,
      and the second supremum involves furthermore an operator with $-\1 \leq T \leq \1$.
      
    \item For any Hermitian-preserving superoperator $\Delta$ and any cptp map
      $T$ on a system $C$, with Hamiltonian grounded $H_C \geq 0$,
      $\| \Delta \|_{\diamond E} = \| \Delta \ox T \|_{\diamond E}$,
      where on the right hand side the energy bound is understood with
      respect to the Hamiltonian $H = H_A \ox \1_C + \1_A \ox H_C$.

    \item For any two Hermitian-preserving superoperators $\Delta_1$
      and $\Delta_2$ acting on $A_1$ and $A_2$, with Hamiltonians $H_{A_1}$ and $H_{A_2}$,
      respectively,
      \[
        \| \Delta_1 \ox \Delta_2 \|_{\diamond E} 
           \geq \max_{E=E_1+E_2} \| \Delta_1 \|_{\diamond E_1} \cdot \| \Delta_2 \|_{\diamond E_2},
      \]
      where the composite system carries the Hamiltonian $H = H_{A_1}\ox\1_{A_2} + \1_{A_1}\ox H_{A_2}$.
  \end{enumerate}
\end{lemma}
\begin{proof}
Straightforward.
\end{proof}

\medskip
Furthermore,
Shirokov~\cite[Prop.~3]{Shirokov:E-diamond} showed that convergence 
w.r.t.~$\|\cdot\|_{\diamond E}$ implies strong convergence, and that
for a discrete Hamiltonian $H_A$ the two topologies are equivalent.

\begin{lemma}
  \label{lemma:SDP}
  The energy-constrained diamond norm is given by a semidefinite programme
  (albeit an infinite dimensional one):
  \[
    \| \Delta \|_{\diamond E} = \max \tr J V \ \text{ s.t. }\ -\1\ox\rho \leq V \leq \1\ox\rho,\ 
                                                               \tr \rho = 1,\ \tr \rho H \leq E,
  \]
  where $J = (\Delta \ox \id)\Phi$, and $\ket{\Phi} = \sum_n \ket{n}\ket{n}$
  is the formal maximally entangled state, with the property that its partial
  trace is $\1$. The computational basis here and in the following 
  is the eigenbasis of $H$.
  
  In the special case of $\Delta = \cN_1-\cN_2$, the difference between two
  cptp maps (equivalently, $\Delta$ is trace-annihilating), this can be simplified to
  \[\begin{split}
    \frac12 \|\Delta\|_{\diamond E} 
                            &= \max \tr J W \ \text{ s.t. }\ 0 \leq W \leq \1 \ox \rho,\ 
                                                                \tr \rho = 1,\ \tr \rho H \leq E \\
                            &= \min x + yE \ \text{ s.t. } \ Z \geq J,\ Z \geq 0,\ 
                                                                x,y \geq 0, x\1 + yH \geq \tr_A Z.
  \end{split}\]
\end{lemma}

\begin{proof}
Basically, the reasoning is precisely that of Watrous \cite{Watrous:diamond},
with the added energy constraint. The derivation of the dual has one more
Lagrange multiplier variable to keep track of the energy bound.

The biggest issue is that $\Phi$ is only a singular operator, and hence $J$,
too. Note however, that to make sense of the traces in the primal SDPs, 
it is enough to have it as a singular state, i.e.~as the limit of, say,
the sequence $\Phi_n=\proj{\Phi_n}$, $\ket{\Phi_n} = \sum_{k=1}^n \ket{k}\ket{k}$
($n\rightarrow\infty$), or of unnormalised entangled squeezed states 
$\sum_k \lambda^k \ket{k}\ket{k}$ ($\lambda\nearrow 1$).
Apart from that, the reasoning for strong duality and hence equality of
primal and dual optimal values, is the same.
\end{proof}

\medskip
As the usual diamond norm, the energy-constrained version has an 
operational interpretation in terms of discriminating channels 
with test states of bounded energy:
$\frac12 \bigl( 1-\| p\cN_1 - (1-p)\cN_2 \|_{\diamond E} \bigr)$ is 
the minimum error probability of distinguishing $\cN_1$ from $\cN_2$, 
with prior probabilities $p$ and $1-p$, respectively, 
when we are allowed preparation of a probe state $\rho^{AC}$ with energy
of $\rho^A$ bounded by $E$, one application of the unknown channel, 
and an arbitrary measurement on the system $BC$.

Before moving on, we pause to mention the closely related notion of
energy-constrained channel fidelity and Bures distance, explored
in \cite{SWAT} and \cite{Shirokov:E-fidelity}:
\[
  \beta_E(\cN_1,\cN_2) = \sup \sqrt{1- F\bigl( (\cN_1\ox\id_C)\rho, (\cN_2\ox\id_C)\rho \bigr)^2}
                                                    \text{ s.t. } \rho \text{ state on }A\ox C,\ 
                                                                  \tr\rho^A H_A \leq E.
\]
It relates to $\|\cN_1-\cN_2\|_{\diamond E}$ via inequalities analogous 
to the relation between Bures distance of states and their trace distance.
Notable, in \cite[Prop.~1]{Shirokov:E-fidelity} it is shown that 
$\beta_E(\cN_1,\cN_2)$ is the minimum $\beta(\mathcal{V}_1,\mathcal{V}_2)$ over all
isometric Stinespring dilations $\mathcal{V}_i$ of $\cN_i$.

\medskip
\begin{remark}
One might wonder why we only impose the energy constraint at the channel
input, and indeed, from a fundamental perspective of paying attention 
to the complete physical resources, it will be more natural to consider
not only the energy required to prepare the test state, but also the energy 
cost of making the measurement at the end.
Navascu\'es and Popescu~\cite{NavascuesPopescu:vision} have done just that,
considering the energetic cost of breaking the conservation law of energy 
at the output system, by the use of an explicit reference frame. 
As far as we know, the combined energy cost of the test state and the
reference frame in discrimination has not been investigated. This in itself
is a well-posed mathematical problem, leading however to further question,
such as the one of the reusability of the state after the measurement, and
a suitable notion of amortised cost.

However, in the present context we can state that the somewhat asymmetric
choice of imposing energy bounds on the test states but not on the measurements
is motivated by the particular applications we have in mind.
\end{remark}


\section{Application I: Continuity of one-parameter groups and semigroups}
\label{sec:continuous-time}
We will now show that the unitary time evolution generated by a densely defined
Hamiltonian $H$ is norm-continuous for $\|\cdot\|_{\diamond E}$, with respect to
the same $H$. Furthermore, we will discuss some semigroup examples.

\subsection{The unitary group $\mathbf{e^{-itH}}$}

For the unitary time evolution $U_t = e^{-itH}$ and $\cU(t) = U_t\cdot U_t^\dagger$, 
with an $H$ that has smallest eigenvalue $0$, we would like to characterise 
$\| \cU(t)-\cU(0) \|_{\diamond E}$ as a function of $E$ and $t$.
Rather than trying to calculate this exactly for concrete Hamiltonians, 
say a spin-$\frac12$ with $\sigma_Z$ Hamiltonian, or a harmonic oscillator
(all of which interesting exercises),
we want to find a universal upper bound that only depends
on $t$ and $E$, uniformly for all Hamiltonians, i.e.~we would like to
evaluate
\[
  \sup_H \frac12 \| \cU(t)-\cU(0) \|_{\diamond E} =: s(t,E),
\] 
or at least upper and lower bound it. This function encodes the
speed limit bounds of Margolus and Levitin \cite{ML-speed}. 
In particular,
if $s(t,E)=1$, it means that there exists a state of energy up
to $E$ that is turned into an orthogonal state in time $t$, so
any bound away from $1$ says that the state cannot move ``too far''
in time $t$.
We are interested in showing that $s(t,E) \rightarrow 0$ for every
fixed $E$ and $t \rightarrow 0$. But also for non-zero $t$ are interesting,
because in a certain sense the function $s(t,E)$ is a refined treatment 
of the quantum speed limit. Here we present a simple upper bound.

\begin{theorem}
  \label{thm:norm-speed-limit}
  The unitary time evolution generated by a grounded Hamiltonian $H$ is
  uniformly continuous in $\|\cdot\|_{\diamond E}$ w.r.t.~the same $H$:
  \begin{equation}
    \frac12 \| \cU(t)-\cU(0) \|_{\diamond E} \leq \sqrt[3]{4tE}.
  \end{equation}
\end{theorem}

\begin{proof}
Consider that $\rho$ has energy bounded by $E$, hence if we truncate the 
coherence with energy levels above $\frac{E}{\epsilon}$, 
we introduce an error less than $\sqrt{\epsilon}$: More precisely,
with the projector $P = \left\{ H \leq \frac{E}{\epsilon} \right\}$,
and letting $\rho' := P \rho P$, we have 
$\frac12 \| \rho-\rho' \|_1 \leq \sqrt{\epsilon}$, 
both at time $0$ and at all later times (this is a manifestation of the 
gentle measurement lemma~\cite{Winter:cq-strong,OgawaNagaoka:cq-code}). So,
\begin{equation}\begin{split}
  \label{eq:xx}
  \frac12 \| \rho(t)-\rho(0) \|_1 &\leq \frac12 \| \rho'(t)-\rho'(0) \|_1 + \sqrt{\epsilon} \\
                                  &\leq t \frac{E}{2\epsilon} + \sqrt{\epsilon},
\end{split}\end{equation}
where the second step is obtained as follows:
We can subdivide the interval $[0;t]$ into arbitrarily many equal steps,
and using the triangle inequality and taking the limit (derivative), we get
\[\begin{split}
  \| \rho'(t)-\rho'(0) \|_1 &\leq n \bigl\| \rho'(t/n)-\rho'(0) \bigr\|_1 \\
                            &\rightarrow t \bigl\| (\id\ox G)\rho' \bigr\|_1,
\end{split}\]
with the generator map $G(\sigma) = i[\sigma,H]$.
By virtue of purification, the maximum of the right hand side is
attained on a pure state, and this pure state can be written 
$\varphi=\proj{\varphi}$, with
$\ket{\varphi} = \sum_k \ket{\phi_k}\ket{k},$
where for simplicity we assume a discrete spectrum of $H$, with
normalised eigenstates $\ket{k}$ (of eigenvalue $E_k \leq \frac{E}{\epsilon}$), 
and suitable vectors $\ket{\phi_k}$ in the reference system, such that
\[
  \langle\varphi\ket{\varphi} = \sum_k \langle\phi_k\ket{\phi_k} = 1.
\]
Now, $(\id \ox G)\varphi = i\proj{\varphi}(\1 \ox H) - i(\1 \ox H)\proj{\varphi}$
has rank $2$ and trace $0$, so its eigenvalues are $x$ and $-x$, and
its trace norm is $2x$. To determine x, we make the observation that 
\[
  2x^2 = \tr [(\id \ox G)\varphi]^2
       = \sum_{k\ell} (E_k-E_\ell)^2 \langle\phi_k\ket{\phi_k} \langle\phi_\ell\ket{\phi_\ell}
       = 2(\Delta E)^2,
\]
where $\Delta E$ is the standard deviation of the energy with
respect to the probability distribution $p_k = \langle\phi_k\ket{\phi_k}$.
Thus,
\[
  \bigl\| (\id\ox G)\rho' \bigr\|_1 \leq \bigl\| (\id \ox G)\varphi \bigr\|_1 = 2 \Delta E.
\]
But $\rho'$ has energy between $0$ and $\frac{E}{\epsilon}$, 
so its energy variance is bounded as $\Delta E \leq \frac{E}{2\epsilon}$.

Now, choosing the cutoff in eq.~(\ref{eq:xx})
optimally, namely $\epsilon = \left(\frac12 tE\right)^{2/3}$, we get
\[
  \frac12 \| \rho(t)-\rho(0) \|_1 \leq 2\sqrt[3]{\frac12 tE},
\]
and because the state was arbitrary for the energy constraint,
this concludes the proof.
\end{proof}

\medskip
As a corollary, 
this formulation includes the result of Margolus-Levitin \cite{ML-speed}
that the time a system takes to move an initial state of energy $E$ to
an orthogonal state is lower bounded $\tau^\perp \geq \frac{1}{4E}$.

\subsection{One-parameter semigroups}

Going back to Proposition \ref{prop:attenuators}, it would be nice to bound
the distance between attenuator channels $\cA_\eta$ and $\cA_{\eta'}$, 
with respect to the grounded version of the quantum harmonic oscillator, 
i.e.~the number Hamiltonian $H=\sum_{n=0}^\infty n \proj{n}$. 
We have not yet been able to compute its value exactly, but it is clear 
from the results of Shirokov \cite{Shirokov:E-diamond} that 
$\| \cA_\eta - \cA_{\eta'} \|_{\diamond E}$ goes to $0$ when $\eta' \rightarrow \eta$,
because $\cA_{\eta'} \rightarrow \cA_{\eta}$ strongly and the Hamiltonian
is discrete. Because of the compactness of the complex unit disc, this convergence
must in fact be uniform in $|\eta-\eta'|$.

Similar arguments can be made for general families of Gaussian channels
with continuous description on the level of covariance matrices and displacements. 
It seems however that computing the actual value of the energy-constrained
diamond distance between two Gaussian channels is a difficult problem.

\subsection{Continuous-variable teleportation}

A related family of maps, though not a semigroup, is obtained by
running the well-known continuous-variable teleportation protocol
with asymptotically highly squeezed entangled states \cite{BraunsteinKimble}. 
It can be shown that this sequence of channels converges to the identity
channel with respect to the energy-constrained diamond norm \cite{PLOB}. 
However, as explained in \cite{Shirokov:E-diamond} and more recently in more detail
in \cite{Wilde:CV-teleportation}, the precise norm used to express this 
statement is not that relevant, because the convergence is most elegantly 
stated to be with respect to the strong topology. The strong convergence
is enough in most circumstances of asymptotic analysis, but there might be
situations in which a quantitative assessment of the error made at finite
squeezing is required. Again, we are not aware of tight bounds on the
quality of the approximate teleportation, which is a special instance of
the distance between Gaussian channels mentioned in the previous subsection.

\section{Application II: Continuity of entropic quantities w.r.t.~$\mathbf{\|\cdot\|_{\diamond E}}$.}
\label{sec:continuity-S}
For two channels $\cN_1$, $\cN_2$ and a test state $\rho^{AC}$ with $\tr \rho^A H_A \leq E$, 
we start by proving a lemma to compare the conditional entropies $S(A|C)_{\omega_i}$ of
the two states $\omega_i^{BC} = (\cN_i\ox\id_C)\rho$, along the lines of \cite{Winter-S}.
To do so, we need a trace distance bound between the states, which is obtained
by supposing $\frac12 \| \cN_1-\cN_2 \|_{\diamond E} \leq \epsilon$. 
We will also need an energy bound on $\omega_i^B$, which first of all requires a
grounded Hamiltonian $H_B \geq 0$ also satisfying the Gibbs hypothesis, and
demand $\tr \omega_i^B H_B \leq \widetilde{E}$.
These ingredients directly lead to the following bounds, with the
well-known monotonic and concave function $g(x) = (1+x)\log(1+x)-x \log x$.

\begin{lemma}
  \label{lemma:conditional-entropy}
  Let $\cN_1$ and $\cN_2$ be quantum channels from $A$ to $B$,
  with $\frac12 \| \cN_1-\cN_2 \|_{\diamond E} \leq \epsilon < 1$.
  Then, for any state $\rho^{AC}$ with $\tr\rho^A H_A \leq E$,
  $\omega_i := (\cN_i\ox\id)\rho$, and
  $\tr \omega_i^B H_B \leq \widetilde{E}$,
  it holds
  \begin{equation}
    \bigl| S(B|C)_{\omega_1} - S(B|C)_{\omega_2} \bigr| 
        \leq 6\epsilon' S\bigl( \gamma_B(\widetilde{E}/\delta) \bigr) + 3 g(\epsilon'),
  \end{equation}
  for any $\epsilon < \epsilon' \leq 1$ and $\delta = \frac{\epsilon'-\epsilon}{1+\epsilon'}$.

  Furthermore, we also have the ``trivial'' bound
  \begin{equation}
    \bigl| S(B|C)_{\omega_1} - S(B|C)_{\omega_2} \bigr| \leq 2 S\bigl( \gamma_B(\widetilde{E}) \bigr).
  \end{equation}
\end{lemma}
\begin{proof}
By definition, we get $\frac12 \| \omega_1 -\omega_2 \|_1 \leq \epsilon$,
hence we can apply \cite[Meta-Lemma 17]{Winter-S}, with $\epsilon'$ and
$\delta$ as specified:
\[
  \bigl| S(B|C)_{\omega_1} - S(B|C)_{\omega_2} \bigr| 
        \leq (2\epsilon'+4\delta) S\bigl( \gamma_B(\widetilde{E}/\delta) \bigr) 
              + 2h(\delta) + (1+\epsilon')h\!\left(\frac{\epsilon'}{1+\epsilon'}\right),
\]
with the binary entropy $h(x) = -x\log x - (1-x)\log(1-x)$.
To get our claimed form, we use the upper bounds
$\delta \leq \frac{\epsilon'}{1+\epsilon'} \leq \min\left\{\frac12, \epsilon'\right\}$,
together with the monotonicity of the binary entropy on the interval
$[0,\frac12]$, and the elementary observation that
$(1+x)h\!\left(\frac{x}{1+x}\right) = g(x)$.

The ``trivial'' bound follows from the energy bound on the states
$\omega_i$ and the fact that under that energy bound, the conditional
entropy is between $-S\bigl( \gamma_B(\widetilde{E}) \bigr)$ and 
$+S\bigl( \gamma_B(\widetilde{E}) \bigr)$.
\end{proof}

\medskip
For the rest of the discussion we shall assume that the channels map 
energy-bounded states to energy-bounded states, more precisely
\[
  f_i(E) := \sup \left\{ \tr\cN_i(\rho) H_B \,:\, \tr\rho H_A\leq E \right\} < \infty
\]
for all $E \geq 0$. These functions are evidently non-decreasing and 
(easy to see) concave. Thus, we will actually use an affine linear upper bound
$f_i(E) \leq \alpha E + E_0$ ($i=1,2$), which can be expressed 
concisely as the operator inequality
\[
  \cN_i^*(H_B) \leq \alpha H_A + E_0.
\]
We call channels obeying such a constraint \emph{energy-limited}.

Gaussian channels that satisfy this condition include all
channels with passive linear unitary dilation and energy-bounded environment, 
all displacements, finite squeezers, and compositions of such channels.

\medskip
We have now the ingredients to imitate, with a small variation, the
telescoping trick from \cite{LeungSmith:cont} to prove continuity bounds 
for $n$-letter conditional entropies of states originating from different
tensor power channels. 

On the composite system, we consider by default the sum Hamiltonian 
$H = \sum_{i=1}^n H_{A_i} \ox \1^{\ox [n]\setminus i}$, whose energy is the
sum of all local energy contributions:
$\tr \rho H = \sum_{i=1}^n \tr \rho^{A_i} H_{A_i}$. 
Bounding the average input energy for a channel is customarily 
done with respect to this Hamiltonian.

\begin{lemma}
  \label{lemma:conditional-telescope}
  Let $\cN_1$ and $\cN_2$ be energy-limited quantum channels from $A$ to $B$,
  with $\frac12 \| \cN_1-\cN_2 \|_{\diamond E} \leq \epsilon < 1$.
  Then, for any state $\rho$ on $A^nC$ with 
  $\tr\rho^{A^n} \overline{H}_{A^n} = \frac1n \sum_{i=1}^n \tr \rho^{A_i} H_{A_i} \leq E$,
  and 
  \[
     \omega_n := (\cN_1^{\ox n} \ox \id_C)\rho, \quad
     \omega_0 := (\cN_2^{\ox n} \ox \id_C)\rho,
  \]
  it holds
  \begin{equation}\begin{split}
    \frac1n \bigl| S(B^n|C)_{\omega_n} - S(B^n|C)_{\omega_0} \bigr| 
         &\leq 14\delta(1+\delta) \, S\Bigl(\gamma_B\bigl( (1+3\delta)\widetilde{E}/\delta \bigr) \Bigr)
                                                                + 3 g\bigl( 2\delta(1+\delta) \bigr) \\
         &\leq 28\delta \, S\bigl( \gamma_B(4\widetilde{E}/\delta) \bigr) + 3 g(4\delta),
  \end{split}\end{equation}
  where $\delta = \sqrt{\epsilon}$ and $\widetilde{E} = \alpha E + E_0$.
\end{lemma}
\begin{proof}
We introduce the following ``interpolations'' 
\[
  \omega_i := (\cN_1^{\ox i} \ox \cN_2^{\ox n-i} \ox \id_C)\rho,\ i=0,\ldots,n,
\] 
between $\omega_n$ and $\omega_0$, so that
\[\begin{split}
  S(B^n|C)_{\omega_n} - S(B^n|C)_{\omega_0}
         &= \sum_{i=1}^n S(B^n|C)_{\omega_i} - S(B^n|C)_{\omega_{i-1}} \\
         &= \sum_{i=1}^n S(B_i|CB^{[n]\setminus i})_{\omega_i} - S(B_i|CB^{[n]\setminus i})_{\omega_{i-1}}.
\end{split}\]
The second equality follows from the fact that $\omega_i$ and $\omega_{i-1}$
differ only on $B_i$ (where different channels have been applied), but
have the same reduced state on $C \ox B^{[n]\setminus i}$, hence
\[
  S(\omega_i^C) = S(\omega_{i-1}^C),\quad
  S(\omega_i^{CB^{[n]\setminus i}}) = S(\omega_{i-1}^{CB^{[n]\setminus i}}).
\]
Thus, via the triangle inequality
\[
  \frac1n \bigl| S(B^n|C)_{\omega_n} - S(B^n|C)_{\omega_0} \bigr| 
      \leq \frac1n \sum_{i=1}^n \bigl| S(B_i|CB^{[n]\setminus i})_{\omega_i} 
                                        - S(B_i|CB^{[n]\setminus i})_{\omega_{i-1}} \bigr|.
\]

To bound the individual terms, introduce the local energy contributions
$E_i := \tr \rho^{A_i} H_{A_i}$
and set $\epsilon_i := \frac12 \| \cN_1-\cN_2 \|_{\diamond E_i}$.
Observe that by assumption, $\frac1n \sum_i E_i \leq E$, and so by 
Lemma~\ref{lemma:diamond-E-norm} (Property 3), we have
$\frac1n \sum_i \epsilon_i \leq \epsilon$. 
Furthermore, $\tr \omega_i^{B_i} H_{B_i} \leq \widetilde{E}_i := \alpha E_i + E_0$.
At the same time, $\frac12 \| \omega_i - \omega_{i-1} \|_1 \leq \epsilon_i$,
so we can apply \cite[Meta-Lemma 17]{Winter-S} in the above simplified
form of Lemma~\ref{lemma:conditional-entropy}:
\[
  \bigl| S(B_i|CB^{[n]\setminus i})_{\omega_i} - S(B_i|CB^{[n]\setminus i})_{\omega_{i-1}} \bigr|
               \leq 6\epsilon_i' S\bigl( \gamma_B(\widetilde{E}_i/\delta_i) \bigr) + 3 g(\epsilon_i'),
\]
with $\epsilon_i < \epsilon_i' \leq 1$ and $\delta_i = \frac{\epsilon_i'-\epsilon_i}{1_+\epsilon_i'}$.
We observed that this bound is very bad for ``large'' $\epsilon_i$, because for those
we are forced to have an even larger $\epsilon_i'$ but a relatively ``small''
$\delta_i$. 
For those indices $i$, we use the trivial bound from Lemma~\ref{lemma:conditional-entropy}
as the better upper bound.

We are thus motivated to partition the index set $[n] = S \stackrel{.}{\cup} L$
into the $i$ with small and large value of $\epsilon_i$, respectively. 
Concretely, with a parameter $t$ to be fixed later, let
\[
  S := \{ i : \epsilon_i \leq t \}, \quad
  L := \{ i : \epsilon_i > t \}.
\]
Note that the set $L$ cannot be too large; indeed, by Markov's inequality,
$\lambda := \frac1n |L| \leq \epsilon/t$.
On the other hand, for $i\in S$, let $\epsilon_i'=2t$
and $\delta_i = \frac{2t-\epsilon_i}{1+2t} \geq \frac{t}{1+2t}$,
so we get
\begin{equation}
  \bigl| S(B_i|CB^{[n]\setminus i})_{\omega_i} - S(B_i|CB^{[n]\setminus i})_{\omega_{i-1}} \bigr|
               \leq 12t\, S\bigl( \gamma_B(\widetilde{E}_i/\delta_i) \bigr) + 3 g(2t),
\end{equation}
For $i\in L$, we will use the trivial bound $2 S\bigl( \gamma_B(\widetilde{E}_i) \bigr)$.

To put together the final estimate, we need to keep track of how much
energy is in the $S$- and $L$-systems, respectively:
\[
  |L|\widetilde{E}_> := \sum_{i\in L} \widetilde{E}_i, \quad
  |S|\widetilde{E}_< := \sum_{i\in S} \widetilde{E}_i.
\]
Observe that $\lambda \widetilde{E}_> + (1-\lambda)\widetilde{E}_< 
= \frac1n \sum_{i=1}^n \widetilde{E}_i \leq \widetilde{E}$, which we shall use
in a minute.
Now,
\[\begin{split}
  \frac1n \bigl| S(B^n|C)_{\omega_n} - S(B^n|C)_{\omega_0} \bigr| 
      &\leq \frac1n \sum_{i\in L} 2 S\bigl( \gamma_B(\widetilde{E}_i) \bigr)
             + \frac1n \sum_{i\in S} \left[ 12t\, S\bigl( \gamma_B((1+2t)\widetilde{E}_i/t) \bigr)
                                                                                 + 3 g(2t)\right] \\
      &\leq 2\lambda\, S\bigl( \gamma_B(\widetilde{E}_>) \bigr)
             + 12t\, S\bigl( \gamma_B((1+2t)\widetilde{E}_</t) \bigr) + 3 g(2t)                     \\
      &\leq 2\lambda\, S\bigl( \gamma_B(\widetilde{E}/\lambda) \bigr) 
             + 12t\, S\!\left( \gamma_B\!\left(\frac{1+2t}{1-\lambda} \widetilde{E}/t\right) \right) 
                                                                                            + 3 g(2t) \\
      &\leq 2\frac{\epsilon}{t}\, S\bigl( \gamma_B(t\widetilde{E}/\epsilon) \bigr) 
             + 12t\, S\!\left( \gamma_B\!\left(\frac{1+2t}{1-\epsilon/t} \widetilde{E}/t\right) \right)
                                                                                            + 3 g(2t),
\end{split}\]
where we have used, in order of appearance: the concavity of the Gibbs state
entropy as a function of the energy, the fact that both $\lambda\widetilde{E}_>$ and
$(1-\lambda)\widetilde{E}_<$ are bounded by $\widetilde{E}$, and
that on the interval $(0,x_0]$, 
$\xi S\bigl( \gamma(F/\xi) \bigr)$ is maximised at $\xi = x_0$ \cite[Cor.~12]{Winter-S}.
Finally, choosing $t=\sqrt{\epsilon}(1+\sqrt{\epsilon})$ concludes the proof.
\end{proof}

\medskip
With this we can prove asymptotic continuity for infinite dimensional channel
capacities, including quantum capacity ($Q$) and classical capacity ($C$),
if the channels are energy-limited and the capacity has a mean-energy 
constraint at the input,
essentially along the lines of the proofs by Leung-Smith \cite{LeungSmith:cont}.
Assume two channels $\cN_1$ and $\cN_2$, both energy-limited as in
Lemma~\ref{lemma:conditional-telescope}.

For instance, the classical capacity with input energy bound $E$ is given by
\begin{align*}
  C(\cN,E)    &= \sup_n \frac1n \chi(\cN^{\ox n},nE), \text{ where}\\
  \chi(\cN,E) &= \sup_{\{p_x,\ket{\psi_x}\in A\}} S(B)_\rho - S(B|X)_\rho
                 \ \text{ s.t. }\ \tr\omega^B H = \sum_x p_x \tr\psi_x H \leq E,
\end{align*}
with respect to the state 
$\rho^{XB} = \sum_x p_x \proj{x}^X \ox \cN(\psi_x)^B 
           = \cN\Bigl( \sum_x p_x \proj{x}^X\ox\psi_x^A \Bigr)$.
Now, the $n$-copy entropies involved are $S(B^n)$ and $S(B^n|X)$, just
as in Lemma~\ref{lemma:conditional-telescope}. We get that 
\[
  \frac1n \bigl| \chi(\cN_1^{\ox n},nE) - \chi(\cN_2^{\ox n},nE) \bigr| 
       \leq 56\delta \, S\bigl( \gamma_B(4\widetilde{E}/\delta) \bigr) + 6 g(4\delta),
\]
with $\delta = \sqrt{\epsilon}$ and $\widetilde{E} = \alpha E + E_0$.
As the right hand side does not depend on $n$, the same bound holds
for $| C(\cN_1,E)-C(\cN_2,E) |$.

Similarly, the quantum capacity with input energy bound $E$ equals
\begin{align*}
  Q(\cN,E)       &= \sup_n \frac1n Q^{(1)}(\cN^{\ox n},nE), \text{ where}\\
  Q^{(1)}(\cN,E) &= \sup_{\ket{\varphi}\in A\ox A'} -S(A|B)_{(\id\ox\Delta)\varphi}
                    \ \text{ s.t. }\ \tr\varphi^{A'} H_{A'} \leq E.
\end{align*}
Here, the telescoping trick is applied to the $n$-copy entropies
$S(B^n)$ and $S(B^n|A)$, because $-S(A|B^n) = -S(A)+S(B^n)-S(B^n|A)$;
once more, using Lemma~\ref{lemma:conditional-telescope} we get
\[
  \frac1n \bigl| Q^{(1)}(\cN_1^{\ox n},nE) - Q^{(1)}(\cN_2^{\ox n},nE) \bigr|
     \leq 56\delta \, S\bigl( \gamma_B(4\widetilde{E}/\delta) \bigr) + 6 g(4\delta).
\]
As the right hand side does not depend on $n$, the same bound holds
for $| Q(\cN_1,E)-Q(\cN_2,E) |$.

We record these results as a quotable theorem. Similar statements can 
be proved for the entanglement-assisted capacity, the private capacity, 
or even the individual capacity; we omit the details here.

\begin{theorem}
  \label{thm:asy-cont}
  Let $\cN_1$ and $\cN_2$ be energy-limited quantum channels from $A$ to $B$,
  equipped with grounded Hamiltonians satisfying the Gibbs hypothesis.
  If $\frac12 \| \cN_1-\cN_2 \|_{\diamond E} \leq \epsilon < 1$, then
  \begin{align*}
    |C(\cN_1,E)-C(\cN_2,E)| &\leq 56\delta\,S\bigl( \gamma_B(4\widetilde{E}/\delta) \bigr) + 6g(4\delta),\\
    |Q(\cN_1,E)-Q(\cN_2,E)| &\leq 56\delta\,S\bigl( \gamma_B(4\widetilde{E}/\delta) \bigr) + 6g(4\delta),
  \end{align*}
  with $\delta=\sqrt{\epsilon}$ and $\widetilde{E} = \alpha E + E_0$.
  \qed
\end{theorem}

\medskip
Note that for $C(\cN,E)$, Shirokov~\cite{Shirokov:E-diamond} has already
proved a uniform continuity bound that is in fact tighter than the above.
The reason that we have nevertheless presented our analysis is that it
yields a flexible tool of general applicability, whereas the reasoning 
in \cite[Props.~6 \&{} 7]{Shirokov:E-diamond} seems to be geared towards 
the classical capacity.

\section{Discussion}
\label{sec:discussion}
The energy-constrained diamond norm is the solution to a distinguishability
problem with constraints, and has found already applications in the theory
of quantum channel capacities and the characterisation of convergence of
channels, and even in quantitative derivations of ``quantum Darwinism''
in infinite dimension \cite{CV-darwinism}.

In the present paper, we have collected some new and some of the previously 
noted properties of this norm, the most important a general tool to prove
uniform continuity of infinite dimensional channels subject to an average
energy constraint at the sender.
As an amusing application to one-parameter unitary groups, we re-derived 
the Margolus-Levitin quantum speed limit in the form of a uniform continuity
of the time evolution with respect to the energy-constrained diamond norm.
Similar statements can be expected for certain concrete one-parameter
semigroups, for which the outstanding problem is the characterisation of
broad classes of generators leading to norm continuous time evolution.
Two important questions in this context are: 
For a given strongly continuous one-parameter group $e^{-itH}$, which are the
Hamiltonians making the group norm continuous w.r.t.~$\|\cdot\|_{\diamond E}$?
For a given one-parameter semigroup, which ones are the Hamiltonians
making the semigroup norm continuous w.r.t.~$\|\cdot\|_{\diamond E}$?

One motivation of the present work, not yet discussed here, is the extension
to continuous variables of the theory of approximate degradable channels
and the (private and quantum) capacity bounds derived in \cite{approx-degrad}.
The results presented here, as well as from \cite{Shirokov:E-diamond}, 
suggest very strongly that using a definition of approximate degradability 
in terms of energy-constrained diamond norms should lead to bounds similar
to \cite{approx-degrad}: For a quantum channel $\cN$ from $A$ to $B$, with
a grounded Hamiltonian $H_A$, and complementary channel $\cN^c$ from $A$ to $E$,
we say that $\cN$ is \emph{$\epsilon$-degradable} if
$\inf_{\cD} \frac12 \| \cN^c - \cD\circ\cN \|_{\diamond E} \leq \epsilon$,
where the minimisation is over all degrading cptp maps $\cD$ from $B$ to $E$.
We leave the investigation of this notion and its application to future work
(see however \cite{SWAT}).

From a foundational point of view on hypothesis testing under constraints \cite{MWW}, 
it would be desirable to unify the constraints on the input energy to the channels 
with the energy constraint at the output measurement as developed in 
\cite{NavascuesPopescu:vision}.

\bigskip
\textbf{Acknowledgments.}
I thank Maksim Shirokov, Krishnakumar Sabapathy, Stefano Mancini,
Stephan Weis, Stefano Pirandola, Gerardo Adesso and Mark Wilde, 
among several other people, for interesting and stimulating 
discussions on diamond norms, energy constraints and applications.
Financial support by the ERC Advanced Grant ``IRQUAT'', 
the Spanish MINECO (grants FIS2013-40627-P and FIS2016-86681-P) 
with the support of FEDER funds,
and the Generalitat de Catalunya CIRIT, project 2014-SGR-966
is acknowledged.

\xpatchcmd\bibsection{19}{2}{}{}
\xpatchcmd\bibsection{\begingroup}{\vskip-15pt\begingroup}{}{}

\end{document}